\documentclass[final,notitlepage,10pt,reqno,graphicx]{amsart}

\usepackage{amssymb,amsthm,amsmath}
\usepackage{url}
\usepackage{mathrsfs}
\usepackage{natbib}
\bibpunct{[}{]}{,}{n}{}{;}

\theoremstyle{plain}
\newtheorem{theorem}{Theorem}

\newtheorem{lemma}[theorem]{Lemma}

\theoremstyle{definition}

\theoremstyle{remark}

\usepackage{amsmath}

\DeclareMathOperator{\D}{d}

\usepackage[german,english]{babel}
\begin{document}

\title{A Functional Identity Involving Elliptic Integrals }
\author{M. Lawrence Glasser and Yajun Zhou}\address{Dpto. de F\'isica Te\'orica, Facultad de Ciencias, Universidad de Valladolid, Paseo Bel\'en 9, 47011 Valladolid, Spain;
Donostia International Physics Center, P. Manuel de
Lardizabal 4, \\ E-20018 San Sebasti\'an, Spain} 
\email{laryg@clarkson.edu}

\address{Program in Applied and Computational Mathematics (PACM), Princeton University, Princeton, NJ 08544; Academy of Advanced Interdisciplinary Sciences (AAIS), Peking University, Beijing 100871, P. R. China }
\email{yajunz@math.princeton.edu, yajun.zhou.1982@pku.edu.cn}
\date{\today}
\maketitle

\begin{abstract}
    We show that the following double integral
\[\int_{0}^\pi\mathrm{d}\, x\int_0^x\mathrm{d}\, y\frac{1}{\sqrt{1-\smash[b]{p}\cos x}\sqrt{1+\smash[b]{q\cos y}}}\]remains invariant as one trades the parameters $p$ and $q$ for $p'=\sqrt{1-p^2}$ and $q'=\sqrt{1-q^2}$ respectively. This invariance property is suggested from symmetry considerations in the operating characterstics of a semiconductor Hall-effect device.\textit{Subject Classification (AMS 2010)}: 33E05\ (Primary), 78A35 (Secondary)\end{abstract}

\section{Introduction}
When an electron current flows perpendicular to a magnetic field through a conducting medium, the charges are forced to deviate to one side creating an imbalance which results in a measurable electric potential conveying important information about the material. A device based on this, so-called Hall effect, has been studied in detail by Ausserlechner \cite[]{Auss} who has found that its operating features are summed up in the Hall-geometry-factor
\[G(\lambda_f,\lambda_p)=\frac{1}{{\bf K'}\left(\frac{1-p}{1+p}\right){\bf K}\left(\frac{1-f}{1+f}\right)}\int_0^1\frac{\int _0^x\frac{\D y}{\sqrt{1-\left(\frac{1-p}{1+p}\right)^2(1-y^2)}\sqrt{1-y^2}}}{\sqrt{1-x^2}\sqrt{1-\left[1-\left(\frac{1-f}{1+f}\right)^2\right](1-x^2)}}\D x.\]Here $p$ and $f$ are related to the input and output resistances  by
$\lambda_f=2{\bf K}(f)/{\bf K'}(f)$ and $\lambda_p={\bf K'}(p)/[2{\bf K}(p)]$, with the complete elliptic integral of the first kind being defined by\[
\mathbf K(\sqrt{t}):=\int_0^{\pi/2}\frac{\D\theta}{\sqrt{1-t\sin^2\theta}}\equiv \mathbf K'(\sqrt{1-t}).
\]Due to the symmetry of the device $G(\lambda_f,\lambda_p)/\sqrt{\lambda_f\lambda_p}$
must be unchanged under the substitution $(\lambda_f,\lambda_p)\rightarrow(2/\lambda_f,2/\lambda_p).$ This can be recast into the remarkable identity that
$$\int_0^{\pi}\frac{\D x}{\sqrt{1-p\cos x}}\int_0^x\frac{\D y}{\sqrt{1+q\cos y}}$$
is invariant under $(p,q)\rightarrow(\sqrt{1-p^2},\sqrt{1-q^2})$, which is our aim to prove in this note.

\section{A Double Integral Identity}
\begin{theorem}\label{thm:pq_recip}For parameters $p,q\in(0,1)$, define correspondingly $p'=\sqrt{1-p^2},q'=\sqrt{1-q^2}$, then we have an integral identity $A(p,q)=A(p',q')$, where\begin{align}
A(p,q):={}&\int_{0}^\pi\D x\int_0^x\D y\frac{1}{\sqrt{1-\smash[b]{p}\cos x}\sqrt{1+\smash[b]{q\cos y}}}\notag\\={}&\frac{4}{\sqrt{(1-\smash[b]p)(1+\smash[b]{q})}}\int_0^{\pi/2}\frac{\D\theta}{\sqrt{1+\frac{2p}{1-p}\sin^2\theta}}\int_0^\theta\frac{\D\phi}{\sqrt{1-\frac{2q}{1+q}\sin^2\phi}}.
\end{align}\end{theorem}

Before proving the functional equation stated in the theorem above, we need to convert double integrals like $A(p,q)$ into single integrals over the products of elliptic integrals and elementary functions, as described in the lemma below.\begin{lemma}\label{lm:K_int_repn}For $0<\beta<\alpha<1$, the following identity holds:\footnote{The constraint  $0<\beta<\alpha<1$ is needed in the derivation of \eqref{eq:K_int_sum}, the validity of which extends to $ \alpha=2p/(p-1)<0,\beta=2q/(1+q)\in(0,1)$, by virtue of analytic continuation. }\begin{align}
&\int_0^{\pi/2}\frac{\D\theta}{\sqrt{1-\smash[b]{\alpha}\sin^2\theta}}\int_0^\theta\frac{\D\phi}{\sqrt{1-\beta\sin^2\smash[b]{\phi}}}\notag\\={}&\frac{1}{\pi}\int_0^\beta\frac{\mathbf K(\sqrt{1-\smash[b]{\beta}})\mathbf K(\sqrt{t})}{\sqrt{1-t}+\sqrt{1-\alpha}}\frac{\D t}{\sqrt{1-t}}+\frac{1}{\pi}\int_\beta^1\frac{\mathbf K(\sqrt{\smash[b]{\beta}})\mathbf K(\sqrt{1-t})}{\sqrt{1-t}+\sqrt{1-\alpha}}\frac{\D t}{\sqrt{1-t}},\label{eq:K_int_sum}
\end{align}where the integrations are carried out along straight line-segments joining the end points.  \end{lemma}
\begin{proof}In what follows, we write $\mathbb Y_\lambda(X):=\sqrt{X(1-X)(1-\lambda X)} $ for $X\in(0,1) $ and $\lambda\in(0,1)$, with the square root taking positive values. It is clear that the complete elliptic integral $\mathbf K(\sqrt{\lambda}),\lambda\in(0,1)$   satisfies \begin{align}
\mathbf K(\sqrt{\lambda})=\frac{1}{2}\int_0^1\frac{\D X}{\mathbb Y_\lambda(X)}.\label{eq:K_int_repn}
\end{align}

For $0<\beta<\alpha<1$, we have an addition formula of Legendre type \cite[][Eq.~2.3.26]{AGF_PartII} \begin{align}
\frac{\pi}{\mathbb Y_{\alpha}(  U)}\int^1_{  U}\frac{\D u}{\mathbb Y_{\beta}(u)}={}&\int_{0}^1 \frac{2\alpha\mathbf K(\sqrt{1-\smash[b]{\beta}})}{1-\alpha  UV} \frac{V\D V}{\mathbb Y_{\alpha}(V)}+\int_{0}^1 \frac{2\alpha\mathbf  K(\sqrt{\smash[b]{\beta}})}{1-(1-\alpha  U)V} \frac{V\D V}{\mathbb Y_{1-\alpha}(V)}\notag\\{}&-\int^{1}_{\frac{1-\alpha}{1-\beta}}\frac{\D X}{\mathbb Y_{1-\beta}(X)}\int_{\frac{1-(1-\beta)X}{\alpha}}^1\frac{\D V}{\mathbb Y_{\alpha}(V)}\frac{\alpha V}{1-\alpha   UV}.\label{eq:3rd_kind_add_comb_prep}
\end{align} Integrating over $U\in(0,1)$, we obtain\begin{align}&
\pi\int_0^1\frac{\D U}{\mathbb Y_{\alpha}(  U)}\int^1_{  U}\frac{\D u}{\mathbb Y_{\beta}(u)}=4\pi\mathbf  K(\sqrt{\smash[b]{\vphantom\beta\alpha}})\mathbf  K(\sqrt{\smash[b]{\beta}})-\pi\int_0^1\frac{\D U}{\mathbb Y_{\alpha}(  U)}\int^U_{  0}\frac{\D u}{\mathbb Y_{\beta}(u)}\notag\\={}&-2\mathbf K(\sqrt{1-\smash[b]{\beta}})\int_{0}^1  \frac{\log(1-\alpha V)\D V}{\mathbb Y_{\alpha}(V)}+2\mathbf  K(\sqrt{\smash[b]{\beta}})\int_{0}^1  \frac{\log\frac{1-(1-\alpha) V}{1-V}\D V}{\mathbb Y_{1-\alpha}(V)}\notag\\{}&+\int^{1}_{\frac{1-\alpha}{1-\beta}}\frac{\D X}{\mathbb Y_{1-\beta}(X)}\int_{\frac{1-(1-\beta)X}{\alpha}}^1\frac{\D V}{\mathbb Y_{\alpha}(V)}\log(1-\alpha V).\label{eq:Apq_add_form}
\end{align}Here, the  first two single integrals over $V$ can be evaluated in closed form \cite[][Eqs.~2.2.3 and 2.2.4]{AGF_PartII}:\begin{align}
\int_{0}^1  \frac{\log(1-\alpha V)\D V}{\mathbb Y_{\alpha}(V)}={}& \mathbf K(\sqrt{\alpha})\log(1-\alpha),\label{eq:logdn_int}
\\\int_{0}^1  \frac{\log\frac{1-(1-\alpha) V}{1-V}\D V}{\mathbb Y_{1-\alpha}(V)}={}&\pi \mathbf K(\sqrt{\alpha})+ \mathbf K(\sqrt{1-\alpha})\log(1-\alpha),\end{align}while the last double integral satisfies
 \cite[cf.][Eq.~2.3.2]{AGF_PartII}\begin{align}
&\int^{1}_{\frac{1-\alpha}{1-\beta}}\frac{\D X}{\mathbb Y_{1-\beta}(X)}\int_{\frac{1-(1-\beta)X}{\alpha}}^1\frac{\D V}{\mathbb Y_{\alpha}(V)}\log(1-\alpha V)\notag\\={}&\frac{2\mathbf K(\sqrt{1-\smash[b]{\beta}})}{\pi}\int_0^1\frac{(1-\beta U)\D U}{\mathbb Y_\beta(U)}\int_0^1\frac{\D W}{\sqrt{W(1-W)}}\frac{\log(1-\alpha W-\beta(1-W))}{1-[\alpha W+\beta(1-W)]U}\notag\\{}&-\frac{2\mathbf K(\sqrt{\smash[b]{\beta}})}{\pi}\int_0^1\frac{[1-(1-\beta)U]\D U}{\mathbb Y_{1-\beta}(U)}\int_0^1\frac{\D W}{\sqrt{W(1-W)}}\frac{\log(1-\alpha W-\beta(1-W))}{1-[1-\alpha W-\beta(1-W)]U}.
\end{align}

Substituting $W=(1-\beta  U)V/(1-\beta UV)$  such that \begin{align}
\frac{W}{1-W}=\frac{(1-\beta  U)V}{1-V},
\end{align}we obtain\begin{align}&\int
_0^1\frac{(1-\beta U)\D U}{\mathbb Y_\beta(U)}\int_0^1\frac{\D W}{\sqrt{W(1-W)}}\frac{\log(1-\alpha W-\beta(1-W))}{1-[\alpha W+\beta(1-W)]U}\notag\\={}&\int
_0^1\frac{\D U}{ \sqrt{U(1-U)}}\int_0^1\frac{\D V}{\sqrt{V(1-V)}}\frac{\log\left( 1-\alpha+\frac{(\alpha-\beta)(1-V)}{1-\beta U V} \right)}{1-\alpha U V},
\end{align}where\begin{align}&
\frac{\log\left( 1-\alpha+\frac{(\alpha-\beta)(1-V)}{1-\beta U V} \right)-\log(1-\alpha V)}{1-\alpha U V}\notag\\={}&\int_0^\beta\left[ \frac{1}{1-tUV}- \frac{1-\alpha}{(1-t)(1-V)+(1-\alpha)(1-tU)V}\right]\frac{\D t}{t-\alpha}
\end{align}allows us to integrate over $V$ and $U$ in a sequel on the right-hand side, leading to \begin{align}&\int
_0^1\frac{(1-\beta U)\D U}{\mathbb Y_\beta(U)}\int_0^1\frac{\D W}{\sqrt{W(1-W)}}\frac{\log(1-\alpha W-\beta(1-W))}{1-[\alpha W+\beta(1-W)]U}\notag\\={}&2\pi\left[\int_0^\beta\frac{\mathbf K(\sqrt{t})}{t-\alpha}\left( 1-\sqrt{\frac{1-\alpha}{1-t}} \right)\D t+\frac{\mathbf K(\sqrt{\alpha})}{2}\log(1-\alpha)\right].
\end{align}Here, in the last step, we have evaluated \begin{align}&\int
_0^1\frac{\D U}{ \sqrt{U(1-U)}}\int_0^1\frac{\D V}{\sqrt{V(1-V)}}\frac{\log(1-\alpha V)}{1-\alpha U V}\notag\\={}&\pi\int_0^1\frac{\log(1-\alpha V)\D V}{\mathbb Y_\alpha(V)}=\pi \mathbf K(\sqrt{\alpha})\log(1-\alpha)
\end{align}with the aid of \eqref{eq:logdn_int}. Likewise, starting with a variable substitution  $W=[1-(1-\beta)  U]V/[1-(1-\beta )UV]$  such that \begin{align}
\frac{W}{1-W}=\frac{[1-(1-\beta)  U]V}{1-V},
\end{align}  we may compute \begin{align}&
\int_0^1\frac{[1-(1-\beta)U]\D U}{\mathbb Y_{1-\beta}(U)}\int_0^1\frac{\D W}{\sqrt{W(1-W)}}\frac{\log(1-\alpha W-\beta(1-W))}{1-[1-\alpha W-\beta(1-W)]U}\notag\\={}&2\pi\left[ \int_1^\beta\frac{\mathbf K(\sqrt{1-t})}{t-\alpha}\left( 1-\sqrt{\frac{1-\alpha}{1-t}} \right)\D t-\frac{\pi\mathbf K(\sqrt{\alpha})}{2}+\frac{\mathbf K(\sqrt{1-\alpha})}{2}\log(1-\alpha) \right].
\end{align}Thus, the claimed identity is verified.\end{proof}

Exploiting the integral identity in the lemma above, together with some modular transformations of elliptic integrals, we will prove Theorem~\ref{thm:pq_recip}.
\begin{proof}[Proof of Theorem~\ref{thm:pq_recip}]We
recall
that  the Legendre function of the first kind of degree $-1/4$ is defined by \begin{align}&
P_{-1/4}(1-2t):={_2}F_1\left(  \left.\begin{array}{c}
\frac{1}{4},\frac{3}{4} \\
1 \\
\end{array}\right| t\right)\notag\\={}&\frac{1}{\sqrt{2}\pi}\int_0^{1}\left[ \frac{u(1-tu)}{1-u} \right]^{-1/4}\frac{\D u}{1-u},\quad t\in\mathbb C\smallsetminus[1,+\infty).
\end{align}The following relations between $P_{-1/4}$ and the complete elliptic integral $\mathbf K$ are recorded in Ramanujan's notebook \cite[][Chap.~33, Theorems 9.1 and 9.2]{RN5}:\begin{align}
\mathbf K\left( \sqrt{\frac{2q}{1+q}} \right)={}&\frac{\pi}{2}\sqrt{1+\smash[b]{q}}P_{-1/4}(1-2q^2),\label{eq:quarter1}\\\mathbf K\left( \sqrt{\frac{1-q}{1+q}} \right)={}&\frac{\pi}{2}\sqrt{\frac{1+\smash[b]{q}}{2}}P_{-1/4}(2q^2-1),\label{eq:quarter2}
\end{align} which are  provable by standard transformations of the respective hypergeometric functions, provided that  $q\in(0,1)$.

With the information listed in the last paragraph, we see that \begin{align}
A(p,q)={}&\int_0^{2q/(1+q)}\frac{\sqrt{2}P_{-1/4}(2q^{2}-1)\mathbf K(\sqrt{t})}{\sqrt{1-t}\sqrt{1-\smash[b]{p}}+\sqrt{1+\smash[b]{p}}}\frac{\D t}{\sqrt{1-t}}\notag\\{}&+\int_{2q/(1+q)}^1\frac{{2}P_{-1/4}(1-2q^{2})\mathbf K(\sqrt{1-t})}{\sqrt{1-t}\sqrt{1-\smash[b]{p}}+\sqrt{1+\smash[b]{p}}}\frac{\D t}{\sqrt{1-t}}.
\end{align}

On one hand, with $t=4\sqrt{s}/(1+\sqrt{s})^2$ and Landen's transformation \cite[][item~163.02]{ByrdFriedman} \begin{align}
\mathbf K(\sqrt{s})={}&\frac{1}{1+\sqrt{s}}\mathbf K\left( \frac{2\sqrt[4]{s}}{1+\sqrt{s}} \right),\quad 0<s<1,\label{eq:Landen_2}
\end{align}we have\begin{align}&\int_0^{2q/(1+q)}\frac{\mathbf K(\sqrt{t})}{\sqrt{1-t}\sqrt{1-\smash[b]{p}}+\sqrt{1+\smash[b]{p}}}\frac{\D t}{\sqrt{1-t}}\notag\\={}&2\int_{0}^{(1-\sqrt{1-\smash[b]{q}^2})/(1+\sqrt{1-\smash[b]{q}^2})}\frac{\mathbf K(\sqrt{s})}{(1-\sqrt{s})\sqrt{1-\smash[b]{p}}+(1+\sqrt{s})\sqrt{1+\smash[b]{p}}}\frac{\D s}{\sqrt{s}}.
\end{align}On the other hand, it is clear from a substitution $t=1-s$ that \begin{align}&
\int_{2q/(1+q)}^1\frac{\mathbf K(\sqrt{1-t})}{\sqrt{1-t}\sqrt{1-\smash[b]{p}}+\sqrt{1+\smash[b]{p}}}\frac{\D t}{\sqrt{1-t}}\notag\\={}&\int_{0}^{(1-q)/(1+q)}\frac{\mathbf K(\sqrt{s})}{\sqrt{s}\sqrt{1-\smash[b]{p}}+\sqrt{1+\smash[b]{p}}}\frac{\D s}{\sqrt{s}}\notag\\={}&\int_0^{(1-q)/(1+q)}\frac{\sqrt{2}\mathbf K(\sqrt{s})}{(1-\sqrt{s})\sqrt{1-\sqrt{1-\smash[b]{p}^2}}+(1+\sqrt{s})\sqrt{1+\smash[b]{\sqrt{1-\smash[b]{p}^2}}}}\frac{\D s}{\sqrt{s}}.
\end{align}Here, the last equality results from a  pair of  elementary identities for $p\in(0,1)$: \begin{align}
\sqrt{\frac{1+\sqrt{1-\smash[b]{p}^{2}}}{2}}\pm
\sqrt{\frac{1-\sqrt{1-\smash[b]{p}^{2}}}{2}}=\sqrt{1\pm \smash[b]{p}},
\end{align}which are readily verified by squaring both sides.

Therefore, with $p'=\sqrt{1-p^2},q'=\sqrt{1-q^2} $, we have \begin{align}
A(p,q)={}&\int_{0}^{(1-q')/(1+q')}\frac{2\sqrt{2}P_{-1/4}(1-2q'^2)\mathbf K(\sqrt{s})}{(1-\sqrt{s})\sqrt{1-\smash[b]{p}}+(1+\sqrt{s})\sqrt{1+\smash[b]{p}}}\frac{\D s}{\sqrt{s}}\notag\\{}&+\int_{0}^{(1-q)/(1+q)}\frac{2\sqrt{2}P_{-1/4}(1-2q^2)\mathbf K(\sqrt{s})}{(1-\sqrt{s})\sqrt{1-\smash[b]{p'}}+(1+\sqrt{s})\sqrt{1+\smash[b]{p'}}}\frac{\D s}{\sqrt{s}},\label{eq:Apq_Ap'q'}
\end{align}which is evidently equal to $A(p',q')$.
\end{proof}
\noindent
{\bf Acknowledgement}

M.L.G. thanks Udo Ausserlechner (Infinion Technologies) and Michael Milgram (Geometrics Unlimited) for insightful correspondence. Financial support of MINECO (Project MTM2014-57129-C2-1-P) and Junta de Castilla y  Leon (UIC 0 11) is acknowledged.

\end{document}